\documentclass[11pt]{article}

\usepackage{fullpage}
\usepackage{times}

\usepackage[utf8]{inputenc}
\usepackage{amsthm,amsmath}
\usepackage{amssymb,stmaryrd}
\usepackage{mathabx}
\usepackage{array}
\usepackage{xcolor}
\usepackage[lined,boxed,commentsnumbered]{algorithm2e}
\usepackage{comment}
\usepackage{hyperref}
\usepackage{graphicx}
\usepackage{enumerate}

\newcommand{\commentLaurent}[1]{\textcolor{red}{(L: #1)}}
\newcommand{\commentAngelo}[1]{\textcolor{blue}{(Ang: #1)}}
\newcommand{\commentChristos}[1]{\textcolor{orange}{(C: #1)}}

\newtheorem{fact}{Fact}

\allowdisplaybreaks

\title{Minimizing Rosenthal's Potential in Monotone Congestion Games}

\author{
		Vittorio Bil\`o\thanks{Universtiy of Salento, Lecce, Italy} \and
        Angelo Fanelli\thanks{LAMSADE – UMR CNRS 7243 – Universit\'e Paris Dauphine-PSL, Paris, France} \and
        Laurent Gourv\`es\footnotemark[2] \and
        Christos Tsoufis\footnotemark[2] \and
        Cosimo Vinci\footnotemark[1]
        }

\date{}

\def\nat{{\mathbb{N}}} 
\DeclareRobustCommand{\stirling}{\genfrac\{\}{0pt}{}}
\def\minpot{\textsc{Min Potential}}
\def\mincost{\textsc{Min Social Cost}}
\def\algmin{\textsf{AlgMinCost}}
\def\sc{{\textsc{SC}}}

\theoremstyle{definition}

\newtheorem{theorem}{Theorem}
\newtheorem{remark}{Remark}
\newtheorem{corollary}[theorem]{Corollary}
\newtheorem{lemma}[theorem]{Lemma}
\newtheorem{claim}{Claim}
\newtheorem{proposition}{Proposition}
\newtheorem{example}{Example}

\def\RP{\mathbb{R}_{\geq 0}} 
\def\nat{{\mathbb{N}}} 

\def\N{{N}} 
\def\R{{R}} 
\def\S{{S}} 
\def\SS{\mathbf{S}} 
\def\latency{{\ell}} 
\def\cost{{c}} 
\def\congestion{{n}} 
\def\strategy{s} 
 
\def\ss{\mathbf{s}} 
\def\pp{\mathbf{p}} 
\def\ee{\mathbf{e}} 

\def\O{{\cal O}} 

\newcommand{\NP}{\textup{NP-hard}}
\newcommand{\PP}{\textup{P}}

\begin{document}

\maketitle

\begin{abstract} 
Congestion games are attractive  because they can model many concrete situations where some competing entities interact through the use of some shared resources, and also because they always admit pure Nash equilibria which correspond to the local minima of a potential function. We explore the problem of computing a state of minimum potential in this setting. Using the maximum number of resources that a player can use at a time, and the possible symmetry in the players' strategy spaces, we settle the complexity of the problem for instances  having monotone (i.e., either non-decreasing or non-increasing) latency functions on their resources. The picture, delineating polynomial and NP-hard cases, is complemented with tight approximation algorithms.         
\end{abstract}

\section{Introduction}
Congestion games form one of the most studied classes of games in (Algorithmic) Game Theory. They provide a model of competition among $n$ strategic players requiring the use of certain resources in a set of $m$ available ones. Every resource has a cost function, also called latency function in the realm of transportation and routing networks, which only depends on the number of its users (a.k.a. the resource congestion). Given a state of the game in which all players have performed a strategic choice, the cost of a player is defined by the sum of the costs of all the selected resources. 

Congestion games owe their success to the fact that they can model several concrete scenarios of competition, such as traffic networks, scheduling, group formation and cost sharing, to name a few \cite{Vocking06}. At the same time, they possess intriguing and useful theoretical properties. In fact, Rosenthal \cite{RR73} proved that every congestion game admits an exact potential ({\em Rosenthal's potential}): a function defined from the set of states of the game to the reals such that, every time a player performs a deviation from a state to another, the difference in the potential equals the difference of the player's costs in the two states. For finite games, this implies that every sequence of deviations in which a player improves her cost must have finite length and end at a pure Nash equilibrium, which is a state in which no player can improve her cost by changing her strategic choice. Years later, Monderer and Shapley \cite{MS96} complemented this result by showing that every game admitting an exact potential is isomorphic to a congestion game.

Several algorithmic questions pertaining congestion games and their notable variants have been posed and addressed in the literature. Among these are computing a Nash equilibrium \cite{FPT04}, computing a state minimizing the sum of the players' costs (a.k.a. the {\em social optimum}) \cite{MS12,PG21}, bounding the worst-case (price of anarchy \cite{KP99}) and the best-case (price of stability \cite{ADKTWR08}) approximation of the social optimum yielded by a pure Nash equilibrium, and computing a state minimizing Rosenthal's potential \cite{delpia2016totally,FPT04,KS21}.

The latter problem, in particular, has interesting applications. First, it follows by definition that every local minimum of an exact potential is a pure Nash equilibrium. So, by computing the (global) minimum of Rosenthal's potential, one obtains a pure Nash equilibrium for a given congestion game. Moreover, potential minimizers usually provide a nice approximation of the social optimum and so they may yield pure Nash equilibria whose efficiency is close to or even match the price of stability \cite{CFKKM11,CK05b}. 

Fabrikant {\em et al.} \cite{FPT04} were the first to attack this problem. They show how to compute the minimum of Rosenthal's potential in symmetric network congestion games with non-decreasing latency functions through a reduction to a min-cost flow problem. In network congestion games, resources are edges in a graph and every player wants to select a path connecting a source to a destination; it is symmetric when all players share the same source-destination pair. Ackermann {\em et al.} \cite{ARV08} extended this result to the case in which all players share the same source (or, equivalently, the same destination) only. Del Pia {\em et al.} \cite{delpia2016totally} and Kleer and Sch{\"{a}}fer \cite{KS21} adopt a polyhedral approach to solve the problem under certain structural properties of the players' strategic space, still in the case of non-decreasing latency functions. They assume that the incidence vectors of the strategies of a player are given by the binary vectors in a certain polytope. Del Pia {\em et al.} \cite{delpia2016totally}, in particular, provide a solution for {\em symmetric totally unimodular congestion games}, i.e., for the case in which the matrix defining the common polytope encoding the strategies of all players is totally unimodular. Kleer and Sch{\"{a}}fer \cite{KS21} further generalize the result to the cases in which the polytope obtained by aggregating the polytopes encoding the strategies of each player satisfies two properties named, respectively, {\em integer decomposition property} (IDP) and {\em box-totally dual integrality property} (box-TDI).

In this work, we continue the study of the problem of computing a state minimizing Rosenthal's potential, that we refer to as \minpot, and depart from previous approaches in what follows. First, rather than considering the combinatorial structure of the players' strategy space, we look at the maximum number of resources that a player can use simultaneously. Secondly, besides of the case of non-decreasing latency functions, which is a typical assumption in road and communication networks where congestion has a detrimental effect on the cost of using a resource, we also consider non-increasing functions, which is typical in cost-sharing scenarios \cite{ADKTWR08}.

\paragraph*{Our contribution.} For games with non-decreasing latency functions, we obtain a precise characterization of \minpot\ with respect to the maximum number of resources that a player can use simultaneously (a.k.a. the {\em size}). 
The results, which also depend on whether players' strategy spaces are symmetric or not, are summarized in Table \ref{tab:summ}. 

\renewcommand{\arraystretch}{1.3}
\begin{table}[ht]
\begin{center}
\begin{tabular}{|l||l l l|} 
\hline
& ${\tt size} = 1$ & ${\tt size} = 2$ & ${\tt size} \geq 3$  \\ 
\hline\hline
\text{symmetric} & $\mathcal{O}(nm)$ (Corollary \ref{cor_1})  & $\mathcal{O}(n^3m^3)$ (Theorem \ref{thm:sym:size2}) & $\NP\;$ (Theorem \ref{thm:sym:size3}) \\  
\hline
\text{general} & $\mathcal{O}(n^3m^3)$ (\cite{KS21}) & $\NP\;$ (Theorem \ref{thm:asym:size2}) & $\NP\;$ (Theorem \ref{thm:sym:size3}) \\ 
\hline
\end{tabular}
\caption{\label{tab:summ} Games with non-decreasing latency functions: Summary of the complexity results with respect to both the size 
and the symmetry of the players' strategy space.}
\end{center}
\end{table}

Given the hardness results stated in Theorems \ref{thm:asym:size2} and \ref{thm:sym:size3} (see Table \ref{tab:summ}), we also focus on the computation of good approximate solutions to \minpot. We heavily exploit an approximation algorithm designed by Paccagnan and Gairing \cite{PG21} for the problem of computing the social optimum in congestion games with non-negative, non-decreasing and semi-convex latency functions. The approximation guarantee, which depends in a fairly complicated way on the values of these functions, is proved to be tight, unless P $=$ NP. For polynomial latency functions of maximum degree $d$ (and non-negative coefficients), Paccagnan and Gairing show that the bound simplifies to the $(d+1)$-th Bell number, denoted as $B_{d+1}$. This result is obtained by exploiting (a generalization of) Dobinski's formula \cite{Mansour15}.

We show how their algorithm can be used to provide approximate solutions to \minpot\ as well. This is done by observing that \minpot\ on a congestion game with non-negative and non-decreasing latency functions can be reduced to the problem of computing a social optimum on the same game with perturbed latency functions which are non-negative, non-decreasing and semi-convex. So, Paccagnan and Gairing's algorithm can be applied. For the case of polynomial latency functions of maximum degree $d$, the reduction produces a game whose latency functions are still polynomials of maximum degree $d$. However, the obtained polynomials are quite specific and possibly have negative coefficients. Then Dobinski's formula cannot be directly applied to obtain tight and explicit bounds on the approximation guarantee, or to show that $B_{d+1}$ continues to hold at least as an upper bound.

As our major contribution, we provide a highly non-trivial analysis of this approximation guarantee, by which we derive a precise bound equal to $\Lambda_d:=\sum_{j=0}^d\frac{j+2}{j+1}\stirling{d}{j}$, where $\stirling{d}{j}$ denotes the Stirling numbers of the second kind (see Section~\ref{sec_def} for formal definitions). It is easy to check that, for any $d\geq 1$, $\Lambda_d$ never exceeds $\frac 3 2 B_{d}$, with the inequality being tight only for the case of $d=1$, and that this value is always smaller than $B_{d+1}$. Moreover, given that $B_{d}$ grows asymptotically as $\left(f(d)\right)^d$ with $f(d)=\Theta(d/\ln(d))$ \cite{Berend2000IMPROVEDBO}, it follows that the difference between $B_{d+1}$ and $\Lambda_d$ increases with $d$. A comparison between $B_{d+1}$ and $\Lambda_d$ for small values of $d$ is reported in Table \ref{table_apx}. Last but not least, since the inapproximability result provided by Paccagnan and Gairing holds for any class of latency functions, we immediately obtain that the approximation guarantee of $\Lambda_d$ is tight for \minpot.

\begin{table}[ht]
\begin{center}
\begin{tabular}{|c||c c c c c c c c|} 
\hline
 & $d=1$ & $d=2$ & $d=3$ & $d=4$ & $d=5$ & $d=6$ & $d=7$ & $d=8$ \\ 
\hline\hline
$B_{d+1}$ & $2$ & $5$ & $15$ & $52$ & $203$ & $877$ & $4140$ & $21147$ \\  
\hline
$\Lambda_{d}$ & $1.5$ & $2.84$ & $6.75$ & $19.54$ & $65.92$ & $251.98$ & $1070.21$ & $4981.15$  \\ 
\hline
\end{tabular}
\caption{\label{table_apx} Comparison between the tight approximation guarantee for the problem of minimizing the social cost (equal to $B_{d+1}$) and for  \minpot\ (equal to $\Lambda_d$), when considering polynomial latency functions of maximum degree $d$.}
\end{center}
\end{table}

For games with non-increasing latency functions, \minpot\ shows to be generally harder, see Table \ref{tab3bis}. In fact, while a solution can be easily computed in the case of symmetric games of constant size, the problem becomes NP-hard as soon as we drop the symmetry assumption,  and this holds even if ${\tt size} = 1$, all resources share the same latency function, and players only have two possible strategies. For general latencies, ${\tt size} = 1$, and no specific limit on the number of possible strategies for the players, we show that \minpot\ cannot be approximated to better than $H_n=\Theta (\ln n)$, unless P $=$ NP, and we provide a matching approximation guarantee (Theorem \ref{thmHn}). 

\begin{table}[ht]
\begin{center}
\begin{tabular}{|l||c|} 
\hline
& ${\tt size} = \mathcal{O}(1)$ \\ 
\hline\hline
\text{symmetric} & $m^{\mathcal{O}(1)}$ (Proposition \ref{prop1}) \\  
\hline
\text{general} & \NP\ when ${\tt size} = 1$ (Theorem \ref{thm2}) \\ 
\hline
\end{tabular}
\caption{\label{tab3bis} Summary of the complexity results for constant size games with non-increasing latency functions.}
\end{center}
\end{table}
\paragraph*{Further related work.}

The problem of computing a global minimum of Rosenthal's potential is a specialization of that of computing a local minimum for this function. This problem, which is equivalent to computing a pure Nash equilibrium for a given game, has received quite some attention in the literature of congestion games. However, while for non-decreasing latency functions a series of results \cite{ARV08,CS11,CMS12,EKM07,FPT04,Fotakis10,DBLP:conf/aaai/IeongMNSS05} has provided a fairly complete understanding of the complexity of this problem, much less in known for the case of non-increasing latency functions \cite{AL13,BFMM21,S10}.

Our approximation for \minpot\ with polynomial latency functions is obtained by exploiting an algorithm designed by Paccagnan and Gairing \cite{PG21}. This algorithm uses taxes to force selfish uncoordinated players to implement provably efficient solutions.
The efficiency of taxation mechanisms in congestion games with non-decreasing latency functions has been studied in a series of papers \cite{BV16,CKK10,PG21,P+21,VS20}. In \cite{BV16,CKK10,PG21,P+21}, the aim is to use taxes to led selfish agents towards states with provably good social cost, while, in \cite{VS20}, the authors also consider the objective of minimizing the {\em stretch}: a worst-case measure of the discrepancy between the potential of a pure Nash equilibrium and the optimal social cost. This measure has application in the computation of approximate pure Nash equilibria.

\section{Preliminaries}\label{sec_def}

\noindent{\bf Mathematical definitions.} Given a positive integer $k$, we denote by $[k]$ the set $\{1,2,\ldots,k\}$. Given two integers $d$ and $k$ with $0\leq k\leq d$, the {\em Stirling number of the second kind}, denoted $\stirling{d}{k}$, is the number of ways to partition a set of $d$ elements into $k$ non-empty subsets. As such, they obey the following recursive definition: $\stirling{d+1}{k}=k\stirling{d}{k}+\stirling{d}{k-1}$. Some simple identities involving these numbers that hold essentially by definition are: $\stirling{d}{d}=1$, $\stirling{d}{1}=1$ for every $d\geq 1$, and $\stirling{d}{2}=2^{d-1}-1$. It has been proved, see \cite{RD69}, that $\stirling{d}{k}\geq\frac{k^2+k+2}{2}k^{d-k-1}-1$. Using this lower bound, it is immediate to show that $\stirling{d}{d-2}\geq\frac{d^3-5d^2+10d-10}{2}$, with the right-hand side always increasing in $d$, which implies $\stirling{d}{d-2}\geq 7$ for every $d\geq 4$. For a given $d\geq 0$, the {\em Bell number}, denoted $B_d$, counts the number of possible partitions of a set of $d$ elements. By definition, it immediately follows that $B_d=\sum_{k=0}^d \stirling{d}{k}$.\footnote{Given that $\stirling{d}{0}=0$ for any $d>0$, this identity can be rewritten as $B_d=\sum_{k=1}^d \stirling{d}{k}$, whenever $d>0$.} For two non-negative integers $i$ and $j$, the {\em falling factorial}, denoted $(i)_j$, is defined as $(i)_j:=i\cdot (i-1)\ldots\cdot (i-j+1)=\prod_{k=0}^{j-1}(i-k)$, with the convention that $(i)_j:=0$ for $j=0$.

\bigskip

\noindent{\bf The model.} A \textsc{congestion game} ${\cal G}$ is represented by a tuple $\langle\N, \R, (\S_i)_{i \in \N}, (\latency_r)_{r \in \R} \rangle$.
$\N$ denotes the set of \emph{players} and $\R$ the set of \emph{resources}. 
We assume that both $\N$ and $\R$ are finite and non-empty and define $n:=|N|$ and $m:=|R|$.
Each player $i\in \N$ is associated with a finite and non-empty set of \emph{strategies}  $\S_i \subseteq 2^{\R}$. 
If every strategy in $\S_i$ consists of one resource then we say that $\cal G$ is a \emph{singleton} game.
If all players share the same set of strategies, i.e., $\S_i = \S_j$ for every $i,j\in \N$, then we say that $\cal G$ is a \emph{symmetric} game (in that case, $\S$ denotes the strategy set of all players).  
We denote by ${\tt size}({\cal G})$ the maximum cardinality of any strategy,  i.e., ${\tt size}({\cal G}) := \max_{i\in \N} \max_{s\in \S_i}|s|$. Hence, a singleton game ${\cal G}$ is such that ${\tt size}({\cal G})=1$.  
Every resource $r \in \R$ is associated with a \emph{latency function} $\latency_r: \nat \mapsto \RP$, which maps the number of users of $r$ to a non-negative real. 
We assume that  $\latency_r$ is monotone for all $r \in \R$; we also suppose that  $\latency_r(0) = 0$ and $\latency_r(1) > 0$.  
Function $\latency_r$ is \emph{non-decreasing} (resp. \emph{non-increasing}) when  $\latency_r(h) \geq \latency_r(h')$ (resp., $\latency_r(h) \leq \latency_r(h')$) for every $h > h'\geq 1$. Sections \ref{sec4} and \ref{sec:approx} deal with instances where every latency function is non-decreasing, whereas Section \ref{sec6} is devoted to instances where every latency function is non-increasing.  
We say that $\ell_r$ is {\em polynomial} of maximum degree $d\in \mathbb{N}$ if $\ell_r(x)=\sum_{q=0}^d\alpha_{r,q}x^q$, for some coefficients  $\alpha_{r,0},\ldots, \alpha_{r,d}\geq 0$; it is {\em affine} if it is polynomial of maximum degree 1 and is {\em linear} if it is affine and $\alpha_{r,0}=0$.

The set of \emph{states} of the game is denoted by $\SS:=\S_1 \times \S_2 \times \ldots \times \S_n$. 
The $i$-th component of a state $\ss\in \S$ is the strategy played  by player $i$ in $\ss$ and is denoted by $\ss_i$. 
For every state $\ss$ and resource $r$, we denote by $\congestion_r(\ss)$ the number of players using resource $r$ in $\ss$, i.e., $\congestion_r(\ss) := |\{i\in \N : r\in \ss_i\}|$, 
and we refer to it as the \emph{congestion} of $r$ in $\ss$.
For every state $\ss$, the \emph{cost} incurred by player $i$ in $\ss$ is  $\cost_i(\ss) := \sum_{r\in \ss_i} \latency_r(\congestion_r(\ss))$.
Notice that, by definition of latency, $\cost_i(\ss) > 0$ for every player $i$ and state $\ss$. 

\bigskip

\noindent{\bf Improving moves, potential function and pure Nash equilibria.}
Let us consider a congestion game ${\cal G} = \langle\N, \R, (\S_i)_{i \in \N}, (\latency_r)_{r \in \R} \rangle$.
For every state $\ss \in \SS$, every player $i\in\N$ and every $s\in \S_i$, we denote by $[\ss_{-i}, \strategy]$ the new state obtained from $\ss$ by setting the $i$-th component, that is the strategy of $i$, to $\strategy$ and keeping all the remaining components unchanged, i.e., if $\bar\ss = [\ss_{-i}, \strategy]$ then $\bar\ss_i= \strategy$ and $\bar\ss_j = \ss_j$
for every player $j\neq i$.

A congestion game is a strategic game in which every player $i$ selects $\strategy \in \S_i$ so as to minimize $\cost_i([\ss_{-i}, \strategy])$. 
The transition from $\ss$ to $[\ss_{-i}, \strategy]$ is called a \emph{move} of player $i$ from state $\ss$. 
We say that a transition from $\ss$ to $[\ss_{-i}, \strategy]$ is an \emph{improving move} for $i$ if $\cost_i([\ss_{-i}, \strategy]) < \cost_i(\ss)$. 
We say that a state-valued function $\Gamma : \SS \mapsto \RP$ is an \emph{exact potential function} for the game if $\Gamma(\ss) - \Gamma([\ss_{-i}, \strategy]) =   \cost_i(\ss) - \cost_i([\ss_{-i}, \strategy])$ holds for every $\ss\in \SS$ and $\strategy\in \S_i$.   This means that, in games admitting an exact potential $\Gamma$, if a player $i$ can make a move from $\ss$ to $[\ss_{-i}, \strategy]$ such that $\Gamma(\ss) > \Gamma([\ss_{-i}, \strategy])$, then the move must be improving for $i$, and the decrease in cost $\cost_i(\ss) - \cost_i([\ss_{-i}, \strategy])$ for player $i$ is exactly $\Gamma(\ss) - \Gamma([\ss_{-i}, \strategy])$. Meanwhile, the existence of an improving move from $\ss$ to $[\ss_{-i}, \strategy]$ by player $i$ implies that $\Gamma(\ss) > \Gamma([\ss_{-i}, \strategy])$, and the decrease in potential $\Gamma(\ss) - \Gamma([\ss_{-i}, \strategy])$ is precisely $\cost_i(\ss) - \cost_i([\ss_{-i}, \strategy])$. Therefore, the number of states $|\SS|$ being finite, every maximal sequence of improvement moves leads to a \emph{local minumum} of $\Gamma$, i.e., to a state in which no further improvement move can be performed. 

 
Such a state is called \emph{pure Nash equilibrium}. 
In other words, we say that a state $\ss\in \SS$ is a \emph{pure Nash equilibrium} if, for every player $i\in \N$ and every strategy $\strategy\in\S_i$, we have $\cost_i(\ss) \leq \cost_i([\ss_{-i}, \strategy])$. 
It is well known that $\Phi_{\cal G}(\ss) := \sum_{r\in\R}\sum_{j=0}^{\congestion_r(\ss)}\latency_r(j),$ called the Rosenthal's potential function \cite{RR73}, is an exact potential function for  $\cal G$. 
Notice that, by definition of latency, $\Phi_{\cal G}(\ss) > 0$ holds for every state $\ss$.

\bigskip

\noindent{\bf Problem statement.} In this work, we are interested in the following problem, that we name \minpot: given a congestion game ${\cal G}$, find a state of $\cal G$ minimizing $\Phi_{\cal G}$.  
Another interesting problem in the realm of congestion games is \mincost, which, given a congestion game ${\cal G}$, asks for a state minimizing the {\em social cost} $\sc_{\cal G}(\ss):=\sum_{i \in \N} \cost_i(\ss)$ of $\cal G$, i.e., the sum of the costs of all players.

Given $\rho \ge 1$, a $\rho$-approximate state $\ss$ is a feasible state of ${\cal G}$ for which 
${\sf f}(\ss) \le \rho \, {\sf f}(\ss^*)$ where $\ss^*$ is a global minimizer of function ${\sf f}$, where ${\sf f}$ can be either $\Phi_{\cal G}$ or $\sc_{\cal G}$. A $\rho$-approximation algorithm is a polynomial time algorithm which always outputs a $\rho$-approximate state.

\section{Complexity of \minpot\ with non-decreasing latencies} \label{sec4}

In this section and the following one, we assume that every latency function is non-decreasing. 
We start by considering the complexity of \minpot\ for the basic case of ${\tt size}({\cal G}) = 1$, i.e., the case of singleton congestion games. It is well known that any singleton congestion game can be interpreted also as a network one. Thus, the algorithm by Fabrikant {\em et al.} \cite{FPT04} for symmetric network congestion games can be applied to symmetric singleton congestion games as well. The reduction of Fabrikant {\em et al.} to min-cost flow produces a parallel-link graph with $nm$ edges. Given that the best algorithm for min-cost flow in a graph with $\alpha$ nodes and $\beta$ edges has complexity $\O(\alpha\beta\log\alpha(\beta+\alpha\log\alpha))$ (see Armstrong and Jin \cite{AJ97}), it follows that \minpot\ can be solved in $\O(n^2m^2)$ using approaches from the current state of the art. 

We give a better algorithm exploiting the fact that, in singleton games, any sequence of improving moves has polynomially bounded length. Together with next proposition, showing that, if the game is symmetric, any local minimum of Rosenthal's potential is also a global minimum, it yields an $\O(nm)$ algorithm.





\begin{proposition}\label{prop:sym:size1}
All pure Nash equilibria of a symmetric congestion game $\cal G$ with ${\tt size}({\cal G}) = 1$ have the same potential.
\end{proposition}
\begin{proof}
Let $\ee^*$ be a state with minimum potential. 
Clearly $\ee^*$ is a pure Nash equilibrium. 
Assume by contradiction that there exists another pure Nash equilibrium $\ee$ such that $\Phi_{\cal G}(\ee^*) < \Phi_{\cal G}(\ee)$. 
Let us denote by $C(\ee)$ the set of all pure Nash equilibria obtained from $\ee$ by renaming the players, 
i.e., $C(\ee) = \{\ss\in \SS : n_r(\ss) = n_r(\ee) \text{ for all } r\in \R\}$.  
Observe that all states in $C(\ee)$ are equilibria and have the same potential.
For any 
state $\ss\in \SS$, let us denote by ${\tt over}(\ss) \subseteq \R$ the set of resources whose 
congestion in $\ss$ is strictly larger than the 
congestion 
in $\ee$ (or any other equilibrium in $C(\ee)$) and by ${\tt under}(\ss) \subseteq \R$ the set of resources whose 
congestion 
in $\ss$ is strictly smaller than the 
congestion in $\ee$, i.e., 
${\tt over}(\ss) 
= \{r\in \R : \congestion_r(\ss) > \congestion_r(\ee)\}$ 
and 
${\tt under}(\ss) 
= \{r\in \R : \congestion_r(\ss) < \congestion_r(\ee)\}$.
Notice that, as long as $\ss$ does not belong to $C(\ee)$, 
then both ${\tt over(\ss)}$ and ${\tt under(\ss)}$ are non-empty.\\
Let us consider a sequence $\ee^* = \ss^0, \ss^1, \ldots, \ss^{k} = \ee'$ of $k + 1\geq 2$ states in which $\ee' \in C(\ee)$ and, for every $t \in [0, \ldots, k-1]$,  $\ss^{t+1}$ is obtained from $\ss^{t}$ by a move of player $\pi(t)$ who deviates from a resource in ${\tt over}(\ss^t)$  to a resource in ${\tt under}(\ss^t)$, i.e.,  
$\ss^t_{\pi(t)} \in {\tt over}(\ss^t)$ and $\ss^{t+1}_{\pi(t)} \in {\tt under}(\ss^t)$. 
Notice that this sequence of states is well defined, and that every move in the sequence decreases $\sum_{r\in \R} |\congestion_r(\ee')-\congestion_r(\ss^t)|$ which is a measure of distance between the congestion vector of any member of $C(\ee)$ and the congestion vector of $\ss^t$, i.e., $\sum_{r\in \R} |\congestion_r(\ee')-\congestion_r(\ss^t)| > \sum_{r\in \R} |\congestion_r(\ee')-\congestion_r(\ss^{t+1})|$ for all $t\in [0, \ldots, k-1]$.
Since $\Phi_{\cal G}(\ee^*) < \Phi_{\cal G}(\ee')$, there must exist a time $t$ such that $\Phi_{\cal G}(\ss^{t+1}) > \Phi_{\cal G}(\ss^t)$; let $h < k$ be the first of such time steps, i.e., $\Phi_{\cal G}(\ss^0) = \Phi_{\cal G}(\ss^1) = \ldots = \Phi_{\cal G}(\ss^h) < \Phi_{\cal G}(\ss^{h+1})$. 
Let us assume that $\pi(h)$ at time $h$ is moving from resource $r$ to $r'$, i.e., $r = \ss^h_{\pi(h)}$ and $r' = \ss^{h+1}_{\pi(h)}$. 
By assumption, $\ss^h$ is a state with minimum potential and therefore an equilibrium, while $\ss^{h+1}$ is not an equilibrium -- in fact the move of player $\pi(h)$ from resource $r'$ to $r$, leading from state $\ss^{h+1}$ to $\ss^h$, decreases the potential and hence is an improving move.
Moreover, observe that, since $r \in {\tt over}(\ss^h)$ and $r'\in {\tt under}(\ss^h)$, we  have that  $n_r(\ss^{h+1}) = n_r(\ss^h) - 1 \geq n_r(\ee')$ and $n_{r'}(\ss^{h+1}) = n_{r'}(\ss^h) + 1 \leq n_{r'}(\ee')$. 
Combining these latter observations with the fact that the latencies are non-decreasing and that the move of $\pi(h)$ from $r'$ to $r$ in $\ss^{h+1}$ is an improving move,  
we get that also the move of any player in $\ee'$ from resource $r'$ to $r$ is an improving move, which implies that $\ee'$ is not an equilibrium.
Hence, a contradiction.\end{proof}


\begin{corollary}\label{cor_1}
For every symmetric congestion game $\cal G$ with ${\tt size}({\cal G}) = 1$, \minpot\ can be solved in ${\cal O}(nm)$ time.
\end{corollary}
\begin{proof}
Given Proposition 
\ref{prop:sym:size1}, it follows that any pure Nash equilibrium for a symmetric singleton congestion game is a solution to \minpot. A Nash equilibrium in this setting can be easily computed as follows. Start from the empty state and let players sequentially choose the cheapest resource, given the choices of her predecessors.  
Every single player's decision can be done in ${\cal O}(m)$ time, for a total complexity of ${\cal O}(nm)$. The outcome is a pure Nash equilibrium because every player's decision is a best response with respect to her predecessors. This is also true for the successors. Indeed, if a successor of player $i$, say $j$,  plays the same resource as $i$, then $i$ cannot profitably deviate because both $i$ and $j$ have the same strategy space. If $j$ plays a different resource from $i$, then $i$ cannot profitably move towards $j'$s resource  because the latency functions are non-decreasing.  \end{proof}

Proposition \ref{prop:sym:size1} extends neither to non-symmetric singleton games nor to symmetric games of size two (see Subsection \ref{sec_A1} in the Appendix).

Now, let us shift towards singleton games when the symmetry property is dropped. 
This case is also in P, since it is covered by the results of Kleer and Sch{\"a}fer on \minpot\ for polytopal congestion games satisfying some structural properties (namely, IDP and box-TDI ) \cite[Theorem 3.3]{KS21}. An alternative (and more direct) way to prove this is to reduce the problem to finding a minimum weight perfect matching in a bipartite graph $(V_1 \cup V_2,E)$ such that $V_1=\N \cup D$, with $D$ being a set of $(m-1)\cdot n$ dummy vertices, $V_2=\{(r,\mu) : r\in \R   \mbox{ and } \mu \in [n]\}$, and there is an edge of weight $\latency_r(\mu)$ between $i \in V_1$ and $(r,\mu) \in V_2$ if, and only if, $r \in \S_i$, and there is an edge of weight $0$ for all pair $(d,(r,\mu)) \in D \times V_2$. Thus, the problem can be solved in $\mathcal{O}(n^3m^3)$.

We now move on to the case of ${\tt size}({\cal G}) = 2$ and show that symmetry makes a huge difference here, as it creates a separation between tractable and intractable cases.

\begin{theorem}\label{thm:sym:size2}
For every symmetric congestion $\cal G$ with ${\tt size}({\cal G}) = 2$, \minpot\ can be solved in $\mathcal{O}(n^3m^3)$.
\end{theorem}
\begin{proof}[Proof sketch.]
We exploit a reduction to the problem of computing a Maximum Weight Matching of a given size. The input is a graph $G = (V, E)$, a weight function $w: E \rightarrow \RP$ and a positive integer $q$ such that $G$ admits a matching of size $q$. The problem is to find a matching $M\subseteq E$ of size exactly $q$ which maximizes $w(M) = \sum_{e \in M} w(e)$.  The problem is known to be polynomial time solvable.\footnote{See \cite[Chapter 18.5f]{schrijver2003combinatorial} for bipartite graphs. The case of non-bipartite graphs can be reduced to the traditional maximum weight matching problem by modifying the instance as follows: increase the weight of every edge by a positive constant $Z$ such that every matching of size $k$ has larger weight than any other matching of size $k-1$. Then, add some extra $|V| - 2q$ dummy vertices, and link them with the original vertices with edges whose weight $W$ is big enough so that any maximum weight matching $\mathcal{M}$ in the new graph must include $|V| - 2q$ edges saturating all the dummy vertices. Apart from these  $|V| - 2q$ heavy edges, $2q$ vertices remain to be saturated, which is done with a matching $M \subset \mathcal{M}$ of cardinality $q$ whose edges all belong to the initial graph, and $M$ has maximum weight in the initial graph. }
    
Take a symmetric congestion game ${\cal G} = \langle \N, \R, \S, (\latency_r)_{r \in \R} \rangle$ with  ${\tt size}({\cal G}) = 2$, where $\S$ denotes the strategy space of all players. We can suppose without loss of generality that every strategy in $\S$ consists of exactly two resources. To do so,  
we introduce a fictitious resource $r_0$ (namely, $\R \gets \R \cup \{ r_0 \}$) so that every singleton strategy $\{r_j\} \in \S$  is replaced by $\{ r_0, r_j\}$. 

Now, we can construct an instance $I$ of the matching problem as follows. For every resource $r_i \in \R$, 
we build a set of exactly $n$ vertices $\{ v_i^1, ..., v_i^{n} \}$. Next, for every pair $\{ r_j, r_k \} \in \S$, we construct a complete bipartite graph between $\{ v_j^1, ..., v_j^{n} \}$ and $\{ v_k^1, ..., v_k^{n} \}$. Every edge $(v_j^a, v_k^b)$, where $j, k \neq 0$, has weight equal to $C - \latency_{r_j}(a) - \latency_{r_k}(b)$, where $C \geq 2 \cdot \max_{r \in \R} \latency_r(n)$ and every edge $( v_j^a, v_0^b )$, where $j \neq 0$, has weight equal to $C - \latency_{r_j}(a)$.
(The proofs of the next two claims are deferred to the Appendix -- \ref{Sec:A2.1} and \ref{Sec:A2.2}).

\begin{claim}\label{claim:symm:3}
A state $\ss$ in ${\cal G}$ gives a matching $M$ of size $n$ in $I$ with weight $w(M) = n \cdot C - \Phi_{\cal G}(\ss)$.
\end{claim}
\begin{claim}\label{claim:symm:2}
A matching $M$ of size $n$ in $I$ gives a state $\ss$ in ${\cal G}$ with potential $\Phi_{\cal G}(\ss) \leq n \cdot C - w(M)$.
\end{claim}

Now, the technique is to compute a Maximum Weight Matching $M$ of size $n$ in $I$. Consequently, from Claim  \ref{claim:symm:2} we have a state $\ss$ with potential $\Phi_{\cal G}(\ss) \leq n \cdot C - w(M)$. Next, assume that ${\cal G}$ admits a state $\ss^{*}$ such that $\Phi_{\cal G}(\ss^{*}) < \Phi_{\cal G}(\ss)$. Then, from Claim  \ref{claim:symm:3} we get a matching $M^{*}$ of size $n$ with weight $w(M^{*}) = n \cdot C - \Phi_{\cal G}(\ss^{*})$. However, using the hypothesis, we get that $w(M^{*}) > n \cdot C - \Phi_{\cal G}(\ss) \geq w(M)$, which is a contradiction. 

Concerning time complexity, computing a maximum weight matching of given size reduces to computing a maximum weight matching in a modified graph whose number of vertices is at most doubled. Computing a maximum weight matching is cubic in the number of vertices. Our initial graph having $nm$ vertices, the time complexity of our procedure is $\mathcal{O}(n^3m^3)$.
\end{proof}




When the symmetry property is dropped, \minpot\ becomes intractable when ${\tt size}({\cal G}) = 2$. 

\begin{theorem} \label{thm:asym:size2}
\minpot\ is $\NP$ for congestion games $\cal G$ with ${\tt size}({\cal G}) = 2$ and linear latencies.

\end{theorem}

Finally, we show that hardness of computation extends to even symmetric games as soon as ${\tt size}({\cal G})$ gets equal to three.

\begin{theorem}\label{thm:sym:size3} \minpot\ is $\NP$ for symmetric congestion games $\cal G$ with ${\tt size}({\cal G}) \geq 3$ and linear latencies.
\end{theorem}

Since, for ${\tt size}({\cal G}) \geq 3$, \minpot\ is NP-hard for the symmetric case, it is also NP-hard for the general case. We observe that, if one modifies the latency functions used in the proofs of Theorems \ref{thm:asym:size2} and \ref{thm:sym:size3} to be such that $\ell_r(1)=1$ and 
$\ell_r(h)=M_\rho$ for each $h\geq 2$, where $M_\rho$ is an appropriate large number, then no $\rho$-approximation algorithm for \minpot\ can be proposed, 
unless P $=$ NP.




\section{Approximating \minpot\ with polynomial latencies}\label{sec:approx}
In this section, we show how to achieve an optimal approximation for \minpot, when considering general congestion games with polynomial latency functions of maximum degree $d\in \mathbb{N}$.
To show this, we first exploit an optimal approximation algorithm developed by Paccagnan and Gairing \cite{PG21} for \mincost. This algorithm applies to congestion games with very general latency functions (satisfying mild assumptions only), and the resulting approximation factor is represented as an infinite sum that depends on the values of these functions. Then, by exhibiting the equivalence between \mincost\ and \minpot, we show how to convert the approximation guarantee obtained by Paccagnan and Gairing for \mincost\ into an optimal approximation for \minpot. Then, we specialize the result to the case of polynomial latency functions of maximum degree $d$ and, by exploiting some techniques arising from combinatorics, we achieve an exact quantification of the approximation factor in terms of a weighted finite sum of Stirling numbers of the second kind. In particular, we will show that, for any fixed $\epsilon>0$, $\minpot$ admits a $(\Lambda_d+\epsilon)$-approximation, with
    \begin{equation}\label{lambda_def_0}
        \Lambda_d:=\sum_{j=1}^d\left(\frac{j+2}{j+1}\right)\stirling{d}{j}\in \left[B_d, \frac{3}{2}B_d\right].
    \end{equation}
We point out that most of the algorithmic machinery used to obtain the desired approximation relies on the work of Paccagnan and Gairing, and our careful analysis specializes their results to \minpot\ applied to games with polynomial latency functions. 
Some value of $\Lambda_d$ are provided in Table \ref{table_apx}.

\subsection{Approximation algorithm for \mincost: a quick overview}
For a given $y\in\mathbb{N}$ and a latency function $\tilde{\ell}$, let\footnote{Since in this section we are going to consider a reduction from \minpot\ on a game $\cal G$ to \mincost\ on a game $\tilde{\cal{G}}$, we shall add a ``tilde'' to the notation pertaining the \mincost\ problem on a game $\tilde{\cal{G}}$.}
    \begin{equation}\label{def_rho}
    \rho_{\tilde{\ell}}(y):=\sup_{y\in \mathbb{N}}\frac{\mathbb{E}_{P\sim \text{ Poi}(y)}[P\tilde{\ell}(P)]}{y\tilde{\ell}(y)}=\frac{\sum_{x=0}^\infty x\tilde{\ell}(x)\frac{y^x}{x!e^y}}{y\tilde{\ell}(y)},
    \end{equation}
    with $\text{Poi}(y)$ denoting the Poisson distribution with parameter $y$; furthermore, for a given class of latency functions $\tilde{\mathcal{L}}$, define 
\begin{equation}\label{def_rho_2}
\rho_{\tilde{\mathcal{L}}}:=\sup_{\tilde{\ell}\in \tilde{\mathcal{L}}}\sup_{y\in \mathbb{N}}\rho_{\tilde{\ell}}(y).
\end{equation}    
    For a given (and arbitrarily small) $\epsilon>0$ and a class of latency functions $\tilde{\mathcal{L}}$ which are non-negative, non-decreasing and semi-convex (i.e., such that function $\tilde{g}(x)=x\tilde{\ell}(x)$ is convex for any $\tilde{\ell}\in \tilde{\mathcal{L}}$), the approximation algorithm provided by Paccagnan and Gairing, denoted here as \algmin, guarantees a $(\rho_{\tilde{\mathcal{L}}}+\epsilon)$-approximation to \mincost, when applied to a congestion game $\tilde{\mathcal{G}}$ with latency functions in $\tilde{\mathcal{L}}$. 
Paccagnan and Gairing also show that the obtained approximation is essentially optimal. Indeed, they show that it is NP-hard to approximate $\mincost$ within a factor better than $\rho_{\tilde{\mathcal{L}}}$, when restricting to congestion games with latencies in $\tilde{\mathcal{L}}$, for any class of latency functions $\tilde{\mathcal{L}}$. 
\subsection{\minpot\ versus \mincost} 
Given a latency function $\ell$, let $\tilde{\ell}$ denote the latency function defined as $\ell(x)=\sum_{h=1}^x \ell(h)/x$ and, given a class of latency functions $\mathcal{L}$, let
$\tilde{\mathcal{L}}:=\left\{\tilde{\ell}:\ell\in \mathcal{L}\right\}$; analogously, given a congestion games $\mathcal{G}$ with latency functions $(\ell_r)_{r\in R}$, let $\tilde{\mathcal{G}}$ be the congestion game obtained from $\mathcal{G}$ by replacing each latency $\ell_r$ with $\tilde{\ell}_r$. 

By resorting to the following proposition, that shows the equivalence between \mincost\  and \minpot, we will see how to apply \algmin\ to \minpot\ in order to have the same approximation guranteed for \mincost\ but on a narrowed set of latency functions. 
\begin{proposition}\label{prop:social_cost}
Let $\mathcal{G}$ be a congestion game with latency functions $(\ell_r)_{r\in R}$. Then: (i) the latency functions of $\tilde{\mathcal{G}}$ are non-negative, non-decreasing and semi-convex; (ii) the potential function $\Phi_\mathcal{G}$ of $\mathcal{G}$ coincides with the social cost $\sc_{\tilde{\mathcal{G}}}$ of $\tilde{\mathcal{G}}$.
\end{proposition}

By combining Proposition \ref{prop:social_cost} with the findings of Paccagnan and Gairing, we obtain in polynomial time a  $(\rho_{\tilde{\mathcal{L}}}+\epsilon)$-approximate solution for $\minpot$ as follows: starting from the input congestion game $\mathcal{G}$, we first construct the corresponding congestion game $\tilde{\mathcal{G}}$ (this can be done in polynomial time); then, by applying \algmin\ we obtain a $(\rho_{\tilde{\mathcal{L}}}+\epsilon)$-approximate solution $\ss$ for game  $\tilde{\mathcal{G}}$, w.r.t.  \mincost; finally, as the potential function $\Phi_{\mathcal{G}}$ of $\mathcal{G}$ and the social cost $\sc_{\tilde{\mathcal{G}}}$ of $\tilde{\mathcal{G}}$ have the same value for all states (by Proposition \ref{prop:social_cost}), we have that $\ss$ is also a $(\rho_{\tilde{\mathcal{L}}}+\epsilon)$-approximate solution for game $\mathcal{G}$, w.r.t. \minpot. Furthermore, by the hardness results of Paccagnan and Gairing and the above observations, we have that the obtained approximation is essentially optimal for $\minpot$ (up to the arbitrarily small constant $\epsilon$).

\subsection{Characterization of the approximation factor for polynomial latencies} 
Let $\mathcal{L}_d$ denote the class of polynomial latency functions of maximum degree $d$.
By the above observations, \algmin\ can be adapted to return a $(\rho_{\tilde{\mathcal{L}_d}}+\epsilon)$-approximation to \minpot. However $\rho_{\tilde{\mathcal{L}_d}}$, as it is represented in definition \eqref{def_rho_2} (reported from \cite{PG21}), is defined in terms of an infinite sum, whose exact value can only be approximated and does not allow for a direct quantification of the asymptotic growth of the approximation factor as a function of $d$.

In the following, we show that $\rho_{\tilde{\mathcal{L}_d}}$ coincides with the value $\Lambda_d$ defined in \eqref{lambda_def_0}, and this characterization leads to a simpler and more precise estimation of the approximation ratio. A similar result has been obtained for the \mincost\ problem by Paccagnan and Gairing \cite{PG21}, who showed, by exploiting the Dobinski's formula \cite{Mansour15}, that their tight approximation factor coincides with $B_{d+1}$. However, it seems that their analysis cannot be directly applied to \minpot\ to obtain the same or similar bounds (see Subsection \ref{subsec:fail} of the Appendix for further details). 
\begin{theorem}\label{thm_polynomial}
For any $d\in \mathbb{N}$, we have $\rho_{\tilde{\mathcal{L}_d}}=\Lambda_d\sim B_d\leq \left(\frac{0.792 d}{\ln(d+1)}\right)^d$.
\end{theorem}
\begin{proof}
Fix $d\in \mathbb{N}$. We first observe that $\Lambda_d\sim B_d$ holds by the right-hand part of \eqref{lambda_def_0} and $B_d\leq \left(\frac{0.792 d}{\ln(d+1)}\right)^d$ has been shown in \cite{Berend2000IMPROVEDBO}. Then, in the remainder of the proof we will focus on equality $\rho_{\tilde{\mathcal{L}_d}}=\Lambda_d$.

Let $\ell_d$ denote the monomial latency function defined as $\ell_d(x)=x^d$. We will first show that $\rho_{\{\tilde{\ell}_d\}}=\Lambda_d$, that is, we are restricting the original class of latency function $\mathcal{L}_d$ to the simple monomial function of degree $d$; then, we will generalize this restricted claim to the whole class $\mathcal{L}_d$ of polynomial latency functions of maximum degree $d$.

For any $y\in \mathbb{N}$, define 
\begin{equation}\label{def_lambda_y}
\Lambda_{d}(y):=\frac{\sum_{j\in [d]}\stirling{d}{j}\left(\frac{y^{j+1}}{j+1}+y^j\right)}{\sum_{j\in [d]}\stirling{d}{j}\frac{(y+1)_{j+1}}{j+1}}.
\end{equation}
By exploiting definition \eqref{def_rho}, we have
\begin{equation}\label{def_rho_3}
\rho_{\tilde{\ell}_d}=\frac{\sum_{x=0}^\infty x\tilde{\ell}_d(x)\frac{y^x}{x!e^y}}{y\tilde{\ell}_d(y)}=\frac{\sum_{x=0}^\infty \sum_{h\in [x]} h^d\frac{y^x}{x!e^y}}{\sum_{h\in [y]}h^d}.
\end{equation}
The following lemma provides an alternative representation for the sum of the first $y$ $d$-th powers, in terms of Stirling numbers of the second kind.
\begin{lemma}\label{lem_sum_power_eq}
For any $y\in \mathbb{N}$, we have
$\sum_{h\in [y]} h^d=\sum_{j\in [d]}\stirling{d}{j}\frac{(y+1)_{j+1}}{j+1}.$
\end{lemma}
The equality reported in the above lemma is folklore and the proof can be found, for instance, in \cite{Knuth89}.
The following lemma, together with Lemma \ref{lem_sum_power_eq}, will be used to show that \eqref{def_lambda_y} and \eqref{def_rho_3} are distinct representations for the same number.
\begin{lemma}\label{lem_rho_lambda}
For any $y\in\mathbb{N}$, we have 
$\sum_{x=0}^\infty \sum_{h\in [x]} h^d\frac{y^x}{x!e^y}=\sum_{j\in [d]}\stirling{d}{j}\left(\frac{y^{j+1}}{j+1}+y^j\right).$
\end{lemma}
\begin{proof}[Proof of Lemma \ref{lem_rho_lambda}]
By applying Lemma \ref{lem_sum_power_eq} with $x$ in place of $y$, we have that 
\begin{equation}\label{lem_rho_lambda_eq1}\\
\sum_{x=0}^\infty \left(\sum_{h\in [x]} h^d\right)\frac{y^x}{x!e^y}=\sum_{x=0}^\infty \left(\sum_{j\in [d]}\stirling{d}{j}\frac{(x+1)_{j+1}}{j+1} \right)\frac{y^x}{x!e^y}=\sum_{j\in [d]}\frac{1}{(j+1)e^y}\stirling{d}{j}\sum_{x=0}^\infty (x+1)_{j+1}\frac{y^x}{x!}.
\end{equation}
For any $j,y\in \mathbb{N}$, we have
\begin{equation}\label{lem_rho_lambda_eq6}
\begin{split}
&\sum_{x=0}^\infty (x+1)_{j+1}\frac{y^x}{x!}=y^{j}\sum_{x=0}^\infty \left(\frac{\partial}{\partial y} \right)^{j+1}\left(y\cdot \frac{y^{x}}{x!}\right)=y^{j}\left(\frac{\partial}{\partial y} \right)^{j+1}\left(y\sum_{x=0}^\infty\left(\frac{y^{x}}{x!}\right)\right)\\
&\quad =y^{j}\left(\frac{\partial}{\partial y} \right)^{j+1}\left(ye^y\right)=y^{j}\left(ye^y+(j+1)e^y\right)=(j+1)e^y\left(\frac{y^{j+1}}{j+1}+y^j\right),
\end{split}
\end{equation}
where $\left(\frac{\partial}{\partial y} \right)^{k}\left(g(y)\right)$ denotes the $k$-th derivative of  $g$, the second and the third equality hold since the series of functions $\sum_{x=0}^\infty f_x$, with $f_x(t)= t\left(\frac{t^{x}}{x!}\right)$ for any $t\geq 0$ and $x\in \mathbb{N}$, uniformly converges to function $f(t)=te^t$ on non-negative closed intervals, and then the series of the derivatives converges to the derivative of the series (this last property is folklore and a proof is given, for instance, in \cite{rudin}).
Finally, by applying \eqref{lem_rho_lambda_eq6} to \eqref{lem_rho_lambda_eq1}, we get
\begin{equation*}
\begin{split}
&\sum_{x=0}^\infty \sum_{h\in [x]} h^d\frac{y^x}{x!e^y}=\sum_{j\in [d]}\frac{1}{(j+1)e^y}\stirling{d}{j}\sum_{x=0}^\infty (x+1)_{j+1}\frac{y^x}{x!}=\sum_{j\in [d]}\stirling{d}{j}\left(\frac{y^{j+1}}{j+1}+y^j\right),
\end{split}
\end{equation*}
that shows the claim.
\end{proof}

By Lemma \ref{lem_sum_power_eq} and Lemma \ref{lem_rho_lambda} we get
\begin{equation}\label{rho_lambda_main_eq}
\rho_{\tilde{\ell}_d}(y)=\frac{\sum_{x=0}^\infty \sum_{h\in [x]} h^d\frac{y^x}{x!e^y}}{\sum_{h\in [y]}h^d}=\frac{\sum_{x=0}^\infty \sum_{h\in [x]} h^d\frac{y^x}{x!e^y}}{\sum_{j\in [d]}\stirling{d}{j}\frac{(y+1)_{j+1}}{j+1}}=\frac{\sum_{j\in [d]}\stirling{d}{j}\left(\frac{y^{j+1}}{j+1}+y^j\right)}{\sum_{j\in [d]}\stirling{d}{j}\frac{(y+1)_{j+1}}{j+1}}=\Lambda_{d}(y),
\end{equation}
for any $y\in \mathbb{N}$. 
The following lemma shows that, independently of the value of $d$, the maximum over $y$ of $\Lambda_d(y)$ is given by $y=1$.
\begin{lemma}\label{lem_lambda_decreasing}
For any $y\in \mathbb{N}$, we have $\Lambda_d(1)\geq \Lambda_d(y)$. 
\end{lemma}
By putting the above results together, we get
\begin{equation}\label{thm_pol_claim_1}
\rho_{\{\tilde{\ell}_d\}}=\sup_{y\in \mathbb{N}}\rho_{\tilde{\ell}_d}(y)=\sup_{y\in \mathbb{N}}\Lambda_d(y)=\Lambda_d(1),
\end{equation}
where the second and the third equality hold by equality \eqref{rho_lambda_main_eq} and Lemma \ref{lem_lambda_decreasing}, respectively. 

We observe that \eqref{thm_pol_claim_1} shows the claim of the theorem, when restricting the class $\mathcal{L}_d$ to the monomial function $\ell_d(x)=x^d$ only. We will generalize \eqref{thm_pol_claim_1} to the whole class $\mathcal{L}_d$ of polynomial latency functions of maximum degree $d$. To do this, it is sufficient to show that 
\begin{equation}\label{thm_pol_claim_2}
\rho_{\tilde{\ell}^*_d}(y)\leq \Lambda_d(1)
\end{equation}
for any $\ell^*_d\in \mathcal{L}_d$ and $y\in \mathbb{N}$. Indeed, both \eqref{thm_pol_claim_1} and \eqref{thm_pol_claim_2} would imply that $$\Lambda_d(1)=\sup_{\ell_d^*\in \mathcal{L}_d}\rho_{\{{\tilde{\ell}}_d^*\}}=\sup_{{\ell}_d^*\in {\mathcal{L}}_d}\sup_{y\in \mathbb{N}}\rho_{{\tilde{\ell}}_d^*}(y)=\rho_{\tilde{{\mathcal{L}}}_d},$$ that is, the claim.
We give a further lemma.
\begin{lemma}\label{lem_lambda_q}
We have $\Lambda_d(1)<\Lambda_{d+1}(1)$.
\end{lemma}
\begin{proof}[Proof of Lemma \ref{lem_lambda_q}]
As $\stirling{d}{j}\leq j \stirling{d}{j}+\stirling{d}{j-1}=\stirling{d+1}{j}$ for any $j\in [d]$, we have
\begin{equation*}
\Lambda_d(1)=\sum_{j\in [d]}\left(\frac{j+2}{j+1}\right)\stirling{d}{j}\leq \sum_{j\in [d]}\left(\frac{j+2}{j+1}\right)\stirling{d+1}{j}< \sum_{j\in [d+1]}\left(\frac{j+2}{j+1}\right)\stirling{d+1}{j}=\Lambda_{d+1}(1),
\end{equation*}
and the claim of the lemma follows.
\end{proof}
Let us fix an arbitrary $y\in \mathbb{N}$ and a latency function ${\ell}_d^*\in {\mathcal{L}}_d$, that is, $\ell_d^*(x)=\sum_{q=0}^d\alpha_q x^q$, for some coefficients $\alpha_0,\ldots, \alpha_d\geq 0$. Let $\beta_q:=\alpha_q\sum_{h\in [y]}h^q$ for any $q\in [d]\cup\{0\}$. We have
\begin{equation*}
\begin{split}
&\Lambda_d(1)=\max_{q\in [d]\cup\{0\}}\Lambda_q(1)\geq \max_{q\in [d]\cup\{0\}}\Lambda_q(y)=\max_{q\in [d]\cup\{0\}}\rho_{\tilde{\ell}_q}(y)\geq \frac{\sum_{q=0}^d\beta_q\cdot \rho_{\tilde{\ell}_q}(y)}{\sum_{q=0}^d\beta_q}\\
&=\frac{\sum_{q=0}^d\alpha_q\left(\sum_{x=0}^\infty \sum_{h\in [x]} h^q\frac{y^x}{x!e^y}\right)}{\sum_{q=0}^d\left(\alpha_q\sum_{h\in [y]}h^q\right)}=\frac{\sum_{x=0}^\infty \sum_{h\in [x]} \left(\sum_{q=0}^d \alpha_qh^q\right)\frac{y^x}{x!e^y}}{\sum_{h\in [y]}\left(\sum_{q=0}^d\alpha_q h^q\right)}=\frac{\sum_{x=0}^\infty \left(\sum_{h\in [x]} \ell_d^*(h)\right)\frac{y^x}{x!e^y}}{\sum_{h\in [y]}\ell_d^*(h)}\\
&=\frac{\sum_{x=0}^\infty  x\tilde{\ell}_d^*(x)\frac{y^x}{x!e^y}}{y\tilde{\ell}_d^*(y)}=\rho_{\tilde{\ell}_d^*}(y),
\end{split}
\end{equation*}
where the first and the second equality, respectively, follow from Lemma \ref{lem_lambda_q} and  equality \eqref{rho_lambda_main_eq} (applied with $q$ in place of $d$). By the above inequalities, inequality \eqref{thm_pol_claim_2} follows. Finally, because of the above observations, both \eqref{thm_pol_claim_1} and \eqref{thm_pol_claim_2} show the claim. 
\end{proof}
\begin{remark}[A variant of Dobinski's formula]
We observe that Theorem \ref{thm_polynomial} and, in particular, the equality shown in Lemma \ref{lem_rho_lambda}, is of independent interest, as it provides a variant of Dobinski's formula \cite{Mansour15}. Indeed, Dobinski's formula states that the $d$-th Bell number $B_{d}=\sum_{j\in [d+1]}\stirling{d}{j}$ is equal to $\sum_{x=0}^\infty \frac{x^d}{x!e}$, while Lemma~\ref{lem_rho_lambda}, applied with $y=1$, states that $\sum_{j\in [d]}\stirling{d}{j}\left(\frac{j+2}{j+1}\right)=\sum_{x=0}^\infty \frac{\sum_{h\in [x]} h^d}{x!e}.$
\end{remark}

\section{\minpot\ with non-increasing latencies 
} \label{sec6}

This section is devoted to \minpot\ for congestion games having {\em non-increasing} latency functions. We shall see that the situation significantly differs from the non-decreasing case. In particular, Proposition \ref{prop:sym:size1} does not hold for non-increasing   latency functions (cf. Example \ref{ex_non_incr} in the Appendix). See Section \ref{sec_A2} in the Appendix for the missing proofs of this section.

\begin{proposition} \label{prop1} If the game is symmetric and $\S$ denotes the strategy space of every player, then 
\minpot\ can be solved  in  $|\S|$ steps. \end{proposition}

Proposition \ref{prop1}  implies that \minpot\ can be solved in $m^{\mathcal{O}(1)}$ operations when every strategy consists of selecting a constant number of resources. 
By a reduction from \textsc{vertex cover}, the following result states that \minpot\ is hard when the symmetry property is dropped,  even if other parameters of the game are significantly restricted.  

\begin{theorem} \label{thm2} \minpot\ is {NP}-hard, even if all the resources have the same latency function, and all the players only have two singleton strategies. 
\end{theorem}

To conclude, observe that the approximability of \minpot\ when ${\tt size} ({\cal G})= 1$ is similar to the approximability of 
\textsc{set cover} 
(every player is ``covered'' by her selected resource).  In the following theorem, $H_k$ stands for the $k$-th harmonic number $1+1/2+1/3+\cdots +1/k=\Theta(\ln k)$.  

\begin{theorem} \label{thmHn} \minpot\ admits a $H_{n}$-approximation algorithm for singleton 
congestion games. Moreover, the  approximation ratio $H_{n}$ 
is  best possible unless P $=$ NP. 
\end{theorem}

\section{Conclusion}

We have considered the complexity of building a state of minimum potential in congestion games with monotone latency functions. Our results show that the symmetry of the players' strategies, together with the maximum number of resources used simultaneously, plays an important role. 

Although it is long known that, in general, computing a pure Nash equilibrium in a congestion game is PLS-complete \cite{FPT04}, an intriguing question for future work is about the complexity of computing a pure Nash equilibrium (i.e., a {\em local} minimum of Rosenthal's potential instead of a global minimum) in a monotone non-decreasing congestion game with $\texttt{size}=2$ (general strategies) or $\texttt{size}=3$ (symmetric strategies).  The same question is of interest in monotone non-increasing congestion game with $\texttt{size}=2$ (general strategies).   

Natural dynamics like better or best response, starting from any initial state, always converge towards a pure Nash equilibrium in congestion games, and the time convergence is known to be polynomial if the instance is singleton \cite{DBLP:conf/aaai/IeongMNSS05}, or the strategies are bases of a matroid \cite{ARV08}. An interesting question is to bound the worst-case convergence time of these dynamics for (possibly monotone) congestion games with $\texttt{size}=\mathcal{O}(1)$.       

We proposed an approximation algorithm for congestion games with non-increasing latency functions and $\texttt{size}=1$ (Theorem \ref{thmHn}) but it would be interesting to have an approximation for bigger sizes.


\newpage
\bibliographystyle{plain}
\bibliography{references}

\appendix
\section{Appendix}
\subsection{Proposition \ref{prop:sym:size1} does not extend}\label{sec_A1}

The following example shows that Proposition \ref{prop:sym:size1} no longer holds if the symmetry assumption is dropped.  

\begin{example}
    Suppose we have two players $a$ and $b$ and a set of three resources $\R = \{r_1, r_2, r_3\}$. The latency function associated with 
    resource $r_1$ is $\latency_{r_1}(h) = h$, with resource $r_2$ is $\latency_{r_2}(h) = 2h$, and with resource $r_3$ is $\latency_{r_3}(h) = h$. Player $a$ can choose between resources $r_2$ and $r_3$, while player $b$ can choose between resources $r_1$ and $r_3$.

    Consider state $\ee_1$, where player $a$ uses $r_2$ and player $b$ uses $r_3$. For $a$, we have that $\latency_{r_2}(1) = 2$, while, for $b$, we have that $\latency_{r_3}(1) = 1$. State $\ee_1$ is an equilibrium because if player $a$ moves to $r_3$, then $\latency_{r_3}(2) = 2$, which is the same as she had before. Also, if player $b$ moves to $r_1$, then $\latency_{r_1}(1) = 1$ and $b$ does not improve. Observe that $\Phi_{\cal G}(\ee_1) = 1 + 2 = 3$.
    
    Consider another state $\ee_2$, where player $a$ uses $r_3$ and player $b$ uses $r_1$. Then, for $a$, we have that $\latency_{r_3}(1) = 1$ and, for $b$, we have that $\latency_{r_1}(1) = 1$. State $\ee_2$ is also an equilibrium because if player $a$ moves to $r_2$, then $\latency_{r_2}(2) = 2$, which is more than what she had before. Also, if player $b$ moves to $r_3$, then she increases her latency too. We have $\Phi_{\cal G}(\ee_2) = 1 + 1 = 2\neq \Phi_{\cal G}(\ee_1)$.  \qed
    

\end{example}


Similarly, by the following example, Proposition \ref{prop:sym:size1} no longer holds, even for symmetric games, if ${\tt size}({\cal G}) = 2$.


\begin{example} \label{counterexample1}
Suppose we have two players and a set of four resources $\R = \{ r_1, r_2, r_3, r_4 \}$. The latency for each resource is $\latency_{r_1}(h) = h$, $\latency_{r_2}(h) = 4h$, $\latency_{r_3}(h) = 2h$, $\latency_{r_4}(h) = 2h$. The common strategy set consists of every subset of $\R$ of cardinality 2. 

Let $\ee_1$ be the state where player $1$ chooses resources $\{ r_1, r_2 \}$ and player $2$ chooses resources $\{ r_3, r_4 \}$. State $\ee_1$ is a Nash equilibrium. In fact, the cheapest resource available for deviation to player $1$ (either $r_3$ or $r_4$) costs $4$, while the costlier resource used by this player ($r_2$) also costs $4$: this implies that player $1$ has no improving deviations. Similarly, the cheapest resource available for deviation to player $2$ ($r_1$) costs $2$, while the costlier resource used by this player (either $r_3$ or $r_4$) also costs $2$, so that player $2$ has no improving deviations. So, $\ee_1$ is an equilibrium such that $\Phi_{\cal G}(\ee_1) = \sum_{i\in [4]}\latency_{r_i}(1) = 9$. However, there is another Nash equilibrium $\ee_2$, in which player $1$ chooses $\{ r_1, r_3 \}$ and player $2$ chooses $\{ r_1, r_4 \}$. In fact, the cheapest resource available for deviation to player $1$ (either $r_2$ or $r_4$) costs $4$, while the costlier resource used by this player (either $r_1$ or $r_3$) costs $2$, so that player $1$ has no improving deviations. Similarly, the cheapest resource available for deviation to player $2$ (either $r_2$ or $r_3$) costs $4$, while the costlier resource used by this player (either $r_1$ or $r_4$) costs $2$, so that player $2$ has no improving deviations. In this case, we have $\Phi_{\cal G}(\ee_2) = \latency_{r_1}(1)+\latency_{r_1}(2)+ \latency_{r_3}(1)+\latency_{r_4}(1) = 7\neq \Phi_{\cal G}(\ee_1)$. \qed
\end{example}


Notice that the strategies of the instance described in Example  \ref{counterexample1} form a uniform matroid. This means that Proposition \ref{prop:sym:size1} cannot be extended from singleton congestion games to matroid congestion games.

\subsection{Missing proofs of Section \ref{sec4}} \label{sec_A3}

\subsubsection{Proof of Claim \ref{claim:symm:3}}\label{Sec:A2.1}
\begin{proof}
Start with an empty matching $M$. For every player $i$, if $\ss_i = \{ r_j, r_k \}$ then add edge $(v_j^a, v_k^b)$ to $M$, where $a$ and $b$ are the smallest available indices of unmatched vertices in $M$. Every edge $(v_j^a, v_k^b) \in M$ such that $j$ and $k$ are both different from $0$ has weight $C - \latency_{r_j}(a) - \latency_{r_k}(b)$. Every other edge $(v_j^a, v_0^b) \in M$ such that $j \neq 0$ has weight $C - \latency_{r_j}(a)$. Therefore, the total weight of $M$ is $n \cdot C - \sum_{r \in \R} \sum_{l=0}^{\congestion_r(\ss)}\latency_r(l) = n \cdot C - \Phi_{\cal G}(\ss)$.
\end{proof}
    
\subsubsection{Proof of Claim \ref{claim:symm:2}}\label{Sec:A2.2}
\begin{proof}
Start with a matching $M$. A \emph{hole} in $M$ is the following configuration. For some resource index $j$, and some index $a<n$, vertex $v_j^a$ is unmatched whereas $v_j^b$ such that $b>a$ is matched. 
Each time there is a hole, $M$ can be modified in such a way that its total weight does not decrease. Indeed, it suffices to replace the edge $(u,v_j^b) \in M$ with $(u,v_j^a)$. Since the latency functions are non-decreasing, the weight of $(u,v_j^a)$ is not smaller than the weight of $(u,v_j^b)$. 

After repeating this process until there are no holes anymore, we get a new matching $M'$ where we know that $w(M') \geq w(M)$. 

Use $M'$ to construct a state $\ss$ of the game where each edge in $M'$ corresponds to the strategy of some player. More precisely, an edge  $(v_j^a,v_k^b) \in M'$ where $j,k \neq 0$ corresponds to the strategy $\{r_j,r_k\}$. In addition, an edge  $(v_0^a,v_j^b) \in M'$ where $j \neq 0$ corresponds to the strategy $\{r_0,r_j\}$ which is $\{r_j\}$ in the original game without $r_0$.

Therefore, we have created a state whose potential $\Phi_{\cal G}(\ss)$ is, by construction, equal to $n \cdot C - w(M')$. Since $w(M) \leq w(M')$, we get $\Phi_{\cal G}(\ss) \leq n \cdot C - w(M)$.  \end{proof}

\subsubsection{Proof of Theorem \ref{thm:asym:size2}}

\begin{proof} We reduce from the 3DM problem, which is known to be NP-complete \cite{GJ79}. The input is a set $M \subseteq X \times Y \times Z$, where $X, Y, Z$ are disjoint sets having the same number $q$ of elements, and the question is to decide whether $M$ contains a matching, i.e., a subset $M' \subseteq M$ such that $|M'| = q$ and no two elements of $M'$ agree in any coordinate.

Take an instance $I$ of 3DM and construct a congestion game ${\cal G}$ as follows. 

There are $q$ players, corresponding to the elements of $X$. The resource set $\R$ is equal to $Y \cup Z$, and $\latency_r(h)=h$ for all $r\in \R$. For every triplet $(x,y,z) \in M$, player $x$ has $\{y,z\}$ in her strategy space $\S_x$.  Thus, every strategy consists of two resources. 

We claim that $I$ is a \textsc{yes} instance of 3DM if and only if ${\cal G}$ admits a state $\ss$ such that $\Phi_{\cal G}(\ss) = 2q$.

($\Rightarrow$) Let $M'$ be a matching of $I$ and build a state $\ss$ of ${\cal G}$ as follows. For each triplet $(x,y,z) \in M'$, player $x$ plays $\{y,z\}$. Since $M'$ is a matching, every resource is played by a single player, so each of the $|Y\cup Z|=2q$ resources contributes by one unit to the potential of $\ss$, meaning that $\Phi_{\cal G}(\ss) = 2q$.

($\Leftarrow$) Let $\ss$ be a state of ${\cal G}$ such that  $\Phi_{\cal G}(\ss) =  \sum_{r\in \R}\sum_{j=0}^{\congestion_r(\ss)}\latency_r(j)= 2q$. For a resource $r \in Y \cup Z$ used by at least one player in $\ss$, we have that 
$\sum_{j=0}^{\congestion_r(\ss)}\latency_r(j)=1+2+\cdots+\congestion_r(\ss)$,  so $\left( \sum_{j=0}^{\congestion_r(\ss)}\latency_r(j) \right) / \congestion_r(\ss)=1$ when $\congestion_r(\ss)=1$, and $\left( \sum_{j=0}^{\congestion_r(\ss)}\latency_r(j) \right) / \congestion_r(\ss)>1$ when $\congestion_r(\ss)>1$. 
Suppose every used resource $r$ ``charges'' $\left( \sum_{j=0}^{\congestion_r(\ss)}\latency_r(j) \right) / \congestion_r(\ss)$ to each of its $\congestion_r(\ss)$ users. The total sum of the charges is equal to $\Phi_{\cal G}(\ss)$. Every player $x$ plays exactly two resources, so every player is charged $2$ when she is the only user of her resources, otherwise she is charged a quantity strictly larger than $2$. Since there are exactly $|X|=q$ players, the total sum of the charges is $2q$, i.e., $\Phi_{\cal G}(\ss)=2q$, if and only if every player is the only user of her resources under $\ss$. Therefore, the set $M':=\{(x,\ss_x \cap Y,\ss_x \cap Z) : x \in X\}$ constitutes a matching of $I$ when $\Phi_{\cal G}(\ss)=2q$. \end{proof}

\subsubsection{Proof of Theorem \ref{thm:sym:size3}}
\begin{proof} We reduce from the X3C problem, which is known to be NP-complete \cite{GJ79}. The input is a finite set $X$, with $|X| = 3q$ and a collection $C$ of 3-element subsets of $X$, and the question is to decide whether $C$ contains an exact cover for $X$, that is, a sub-collection $C' \subseteq C$ such that every element of $X$ occurs in exactly one member of $C'$.

Take an instance $I$ of X3C and construct a congestion game ${\cal G}$ as follows. There are $q$ players. Each $r \in X$ is associated with a resource $r$ with latency function $\latency_{r}(h)=h$. All the players have the same strategy space $C$. Thus, ${\cal G}$ is a symmetric game where every strategy consists of three resources. 

We claim that $I$ is a \textsc{ yes} instance of X3C if and only if ${\cal G}$ admits a state $\ss$ such that $\Phi_{\cal G}(\ss) = 3q$.

($\Rightarrow$) Let $C'$ be a sub-collection of $C$ covering every element of $X$ exactly once. Thus, $|C'|=q$. In ${\cal G}$, each player $i$ is associated with a distinct 3-set $T_i$ of $C'$, i.e., $\ss_i=T_i$. Therefore, every resource is played by a single player, so each of the $3q$ resources contributes by one unit to the potential, meaning that $\Phi_{\cal G}(\ss) = 3q$.

($\Leftarrow$) Let $\ss$ be a state of ${\cal G}$ such that  $\Phi_{\cal G}(\ss) =  \sum_{r\in X}\sum_{j=0}^{\congestion_r(\ss)}\latency_r(j)= 3q$. For a resource $r \in X$ used by at least one player in $\ss$, we have that 
$\sum_{j=0}^{\congestion_r(\ss)}\latency_r(j)=1+2+\cdots+\congestion_r(\ss)$,  so $\left( \sum_{j=0}^{\congestion_r(\ss)}\latency_r(j) \right) / \congestion_r(\ss)=1$ when $\congestion_r(\ss)=1$, and $\left( \sum_{j=0}^{\congestion_r(\ss)}\latency_r(j) \right) / \congestion_r(\ss)>1$ when $\congestion_r(\ss)>1$. Suppose every used resource $r$ ``charges'' $\left( \sum_{j=0}^{\congestion_r(\ss)}\latency_r(j) \right) / \congestion_r(\ss)$ to each of its $\congestion_r(\ss)$ users. The total sum of the charges is equal to $\Phi_{\cal G}(\ss)$. Every player $i$ plays exactly three resources, so every player is charged $3$ when she is the only user of her resources, otherwise she is charged a quantity strictly larger than $3$. Since there are exactly $q$ players, the total sum of the charges is $3q$, i.e., $\Phi_{\cal G}(\ss)=3q$, if and only if every player is the only user of her resources under $\ss$. Therefore, when $\Phi_{\cal G}(\ss)=3q$, the set $C':=\{\ss_i : i \in \N\}$ constitutes an exact cover of $X$.
\end{proof}

\subsection{Missing proofs of Section \ref{sec:approx}} \label{sec_A4}

\subsubsection{Proof of Proposition \ref{prop:social_cost}}

\begin{proof}
As for claim (i), we first recall that the latency functions of $\mathcal{G}$ are non-negative and non-decreasing by hypothesis. Thus, by exploiting the definition of each $\tilde{\ell}_r$, we have that the latency functions of $\tilde{\mathcal{G}}$ are non-negative and non-decreasing, too. Furthermore, because of the non-decreasing monotonicity of each function $\ell_r$, we also have that each function $\tilde{g}_r$ defined as  $\tilde{g}_r(x):=x\cdot \tilde{\ell}_r(x)=\sum_{h\in [x]}\ell_r(x)$ is convex (where the last equality holds by definition of $\tilde{\ell}_r$), that is, each $\tilde{\ell}_r$ is semi-convex. Thus, claim (i) holds. 

As for claim (ii), let $\ss$ be an arbitrary state (of both $\mathcal{G}$ and $\tilde{\mathcal{G}}$). We have
\begin{align*}
&\Phi_{\mathcal{G}}(\ss)=\sum_{r\in R}\sum_{h=1}^{n_r(\ss)}\ell_r(h)=\sum_{r\in R}n_r(\ss)\cdot \tilde{\ell}_r(n_r(\ss))\\
&\quad =\sum_{r\in R}\sum_{i\in \N:r\in s_i}\tilde{\ell}_r(n_r(\ss))=\sum_{i\in N}\sum_{r\in s_i}\tilde{\ell}_r(n_r(\ss))=\sum_{i\in \N}\tilde{c}_i(\ss)=\sc_{\tilde{\mathcal{G}}}(\ss),
\end{align*}
and, by the arbitrariness of $\ss$, this shows claim (ii). 
\end{proof}

\subsubsection{Why might previous approaches for \mincost\ fail for \minpot?}\label{subsec:fail}
Paccagnan and Gairing, by exploiting Dobinski's formula \cite{Mansour15}, showed that their bound on the approximation guarantee for the \mincost\ problem applied to congestion games with polynomial latency functions of maximum degree $d$ (and non-negative coefficients) simplifies to the $(d+1)$-th Bell number $B_{d+1}$. Considering that, by Bernoulli's formula \cite{Knuth89}, $\tilde{\mathcal{L}}_d$ is a collection of polynomials of maximum degree $d$, it might seem that we could reuse the same reasoning to directly obtain the same bound. However, the coefficients of the polyomials appearing in $\tilde{\mathcal{L}}_d$ are not arbitrary and can be negative. Thus, when dealing with the \minpot\ problem, Dobinski's formula cannot be directly applied as in \cite{PG21} to obtain more explicit forms for the approximation guarantee.

\subsubsection{Proof of Lemma \ref{lem_lambda_decreasing}}
\begin{proof}
We will first show the claim for $d\in \{1,2,3\}$, and then we will move to the general case $y\geq 4$. If $d=1$ and $y\geq 1$ is arbitrary, recalling that $\stirling{1}{1}=1$, we have
\begin{equation*}
\Lambda_{d}(y)=\frac{y^2/2+y}{y^2/2+y/2}=\frac{y/2+1}{y/2+1/2}.
\end{equation*}
As the last term of the above equalities is decreasing in $y\geq 1$, we immediately derive $\Lambda_{d}(y)\leq\Lambda_{d}(1)$
and this shows the claim for $d=1$. 

If $d=2$, recalling that $\stirling{2}{1}=\stirling{2}{2}=1$, for each $y\geq 1$, we have
\begin{equation*}
\Lambda_{d}(y)=\frac{2y^2+9y+6}{(y+1)(2y+1)}.
\end{equation*}
By taking the derivative of $\Lambda_{d}(y)$ w.r.t. to $y\geq 1$, we have
\begin{equation*}
\frac{\partial}{\partial y}\Lambda_{d}(y)=-\frac{12 y^2 + 20 y + 9}{(y + 1)^2 (2 y + 1)^2}.
\end{equation*}
As $12 y^2 + 20 y + 9\geq 0$ holds for any $y\geq 1$, we have that the above derivative is non-positive for any $y\geq 1$, that is, $\Lambda_{d}(y)$ is non-incrreasing in $y\geq 1$. Thus, we conclude that, $\Lambda_{d}(y)\leq \Lambda_{d}(1)$ holds for any integer $y\geq 1$, if $d=2$. 

If $d=3$, recalling that $\stirling{3}{1}=\stirling{3}{3}=1$ and $\stirling{3}{2}=3$, for each $y\geq 1$, we have
\begin{equation*}
\Lambda_{d}(y)=\frac{y^3+8y^2+14y+4}{(y+1)(y^2-y+2)}.
\end{equation*}
By taking the derivative of $\Lambda_{d}(y)$ w.r.t. to $y\geq 1$, we have
\begin{equation*}
\frac{\partial}{\partial y}\Lambda_{d}(y)=-\frac{2(4y^4+13y^3-y^2 -16 y - 12)}{(y + 1)^2 (y^2- y + 2)^2}.
\end{equation*}
As $4y^4+13y^3-y^2 -16 y - 12\geq 0$ holds for any $y\geq 2$, we have that the above derivative is non-positive for any $y\geq 2$, that is, $\Lambda_{d}(y)$ is non-increasing in $y\geq 2$. Thus, since $\Lambda_{d}(2)=6<\frac{27}{4}=\Lambda_{d}(1)$, we conclude that $\Lambda_{d}(y)\leq \Lambda_{d}(1)$ holds for any integer $y\geq 1$, if $d=3$.

In the remainder of the proof, we will consider the general case of $d\geq 4$. Let $N(y)$ and $D(y)$ denote, respectively, the numerator and the denominator of the right-hand side of \eqref{def_lambda_y}. We show the following facts.
\begin{fact}\label{fact1}
$D(y)\geq y^{d+1}\left(\frac{1}{d+1}+\frac{1}{2y}\right)$ for any integer $y\geq 2$. 
\end{fact}
\begin{proof}[Proof of Fact \ref{fact1}]
We will first show that
\begin{equation}\label{fact1_eq1}
h^d\geq \frac{h^{d+1}-(h-1)^{d+1}}{d+1}+\frac{h^d-(h-1)^d}{2},
\end{equation}
holds for any $h\in \mathbb{N}$, by using a geometric argument. Given a measurable set $S\subseteq \mathbb{R}^2$, let $\mu(S)$ denotes its area (in the sense of Lebesgue measure). Let $A:=\{(x,y):x\in [h-1,h],y\in [0,h^d]\}$, $B:=\{(x,y):x\in [h-1,h],y\in [0,x^d]\}$ and $C:=\{(x,y):x\in [h-1,h],y\in [(x-(h-1))h^d+(h-1)^{d},h^d]\}$. We observe that $A$ is a rectangle with height of $h^d$ and base coinciding with the segment on the $x$-axis from $h-1$ to $h$, $B$ is the non-negative subgraph of the function $x^d$ restricted to $x\in [h-1,h]$ and $C$ is the right triangle generated by vertices $v_1=(h-1,(h-1)^d),v_2=(h-1,h^d),v_3=(h,h^d)$. 

We trivially have that $C\subseteq A$; furthermore, since $x^d$ is non-decreasing, we have that $B\subseteq A$. Thus, we necessarily have that $\mu(A)\geq \mu(B\cup C)=\mu(B)+\mu(C)-\mu(B\cap C)$. Now, as $x^d$ is convex, we necessarily have that $\mu(B\cap C)=0$, and then, by continuing from the above inequality, we obtain $\mu(A)\geq \mu(B)+\mu(C)$. Finally, observing that $\mu(A)=h^d$, $\mu(B)=\int_{h-1}^h x^d dx=\frac{h^{d+1}-(h-1)^{d+1}}{d+1}$ and $\mu(C)=\frac{h^d-(h-1)^d}{2}$, we obtain
\begin{equation*}
h^d=\mu(A)\geq \mu(B)+\mu(C)=\frac{h^{d+1}-(h-1)^{d+1}}{d+1}+\frac{h^d-(h-1)^d}{2},
\end{equation*}
and this shows inequality \eqref{fact1_eq1}.

Now, since Lemma \ref{lem_sum_power_eq} states that $D(y)$ is equal to the $\sum_{h\in [y]}h^d$, \eqref{fact1_eq1} implies that
\begin{equation*}
D(y)=\sum_{h\in [y]}h^d\geq \sum_{h\in [y]}\left(\frac{h^{d+1}-(h-1)^{d+1}}{d+1}+\frac{h^d-(h-1)^d}{2}\right)=\frac{y^{d+1}}{d+1}+\frac{y^d}{2},
\end{equation*}
where the last equality holds by trivial telescopicity.
\end{proof}
\begin{fact}\label{fact2}
$N(y)\leq y^{d+1}\left(\sum_{j\in [d]} \stirling{d}{j}\right)\left(\frac{1}{d+1}+\frac{1}{2y}\right)$
for any integer $y\geq 2$.
\end{fact}
\begin{proof}
Given an integer $y\geq 2$, we have
\begin{align}
N(y)&=\sum_{j\in [d]} \stirling{d}{j}\left(\frac{y^{j+1}}{j+1}\right)+\sum_{j\in [d]} \stirling{d}{j}y^j\nonumber\\
&\leq \sum_{j\in [d]}\stirling{d}{j}\left(\frac{y^{d+1}}{d+1}\right)+\sum_{j\in [d]} \stirling{d}{j}y^j\label{fact2_eq1}\\
&=y^{d+1} \left(\sum_{j\in [d]}\stirling{d}{j}\left(\frac{1}{d+1}\right)+\sum_{j\in [d]}\stirling{d}{j}\left(\frac{1}{y^{d+1-j}}\right)\right)\nonumber\\
&= y^{d+1} \left(\left(\sum_{j\in [d]}\stirling{d}{j}\right)\frac{1}{d+1}+\left(\sum_{h=0}^2 \frac{\stirling{d}{d-h}}{y^h}+\sum_{j\in [d-3]}\frac{\stirling{d}{j}}{y^{d-j}}\right)\frac{1}{y}\right)\nonumber\\
&\leq y^{d+1} \left(\left(\sum_{j\in [d]}^{d}\stirling{d}{j}\right)\frac{1}{d+1}+\left(\sum_{h=0}^2 \frac{\stirling{d}{d-h}}{2^h}+\sum_{j\in [d-3]}\frac{\stirling{d}{j}}{8}\right)\frac{1}{y}\right),\label{fact2_eq2}
\end{align}
where \eqref{fact2_eq1} holds since $\frac{y^{j+1}}{j+1}\leq \frac{y^{d+1}}{d+1}$ for any positive integer $j\leq d$ (as function $g(t)=y^t/t$ is increasing in $t\geq 1$, for any fixed $y\geq 2$) and \eqref{fact2_eq2} holds since $y\geq 2$ implies that $y^h\geq 2^h$ for any $h\in \{0,1,2\}$ and $y^{d-j}\geq 8$ for any positive integer $j\leq d-3$.

Now, considering that $\stirling{1}{1}=1$ and $\stirling{d}{d-2}\geq 7$ for $d\geq 4$, we have $\frac{\stirling{d}{d}}{2}=\frac{1}{2}\leq \frac{7}{4}\leq  \frac{\stirling{d}{d-2}}{4}$, which implies
$$
\stirling{d}{d}+\frac{\stirling{d}{d-2}}{4}=\frac{\stirling{d}{d}}{2}+\frac{\stirling{d}{d}}{2}+\frac{\stirling{d}{d-2}}{4}\leq\frac{\stirling{d}{d}}{2}+\frac{\stirling{d}{d-2}}{4}+\frac{\stirling{d}{d-2}}{4}=\frac{\stirling{d}{d}}{2}+\frac{\stirling{d}{d-2}}{2}.
$$
By the above inequality we obtain
\begin{equation}\label{fact2_eq3}
\sum_{h=0}^2 \frac{\stirling{d}{d-h}}{2^h}=\stirling{d}{d}+\frac{\stirling{d}{d-2}}{4}+\frac{\stirling{d}{d-1}}{2}\leq \sum_{h=0}^2 \frac{\stirling{d}{d-h}}{2}.
\end{equation}
By applying \eqref{fact2_eq3} to \eqref{fact2_eq2} we obtain 
\begin{align*}
N(y)&\leq y^{d+1} \left(\left(\sum_{j\in [d]}\stirling{d}{j}\right)\frac{1}{d+1}+\left(\sum_{h=0}^2 \frac{\stirling{d}{d-h}}{2^h}+\sum_{j\in [d-3]}\frac{\stirling{d}{j}}{8}\right)\frac{1}{y}\right)\\
&\leq y^{d+1}\left(\left(\sum_{j\in [d]}\stirling{d}{j}\right)\frac{1}{d+1}+\left(\sum_{h=0}^2 \frac{\stirling{d}{d-h}}{2}+\sum_{j\in [d-3]}\frac{\stirling{d}{j}}{8}\right)\frac{1}{y}\right)\\
&\leq y^{d+1}\left(\left(\sum_{j\in [d]}\stirling{d}{j}\right)\frac{1}{d+1}+\left(\sum_{h=0}^2 \frac{\stirling{d}{d-h}}{2}+\sum_{j\in [d-3]}\frac{\stirling{d}{j}}{2}\right)\frac{1}{y}\right)\\
&\leq y^{d+1}\left(\left(\sum_{j\in [d]}\stirling{d}{j}\right)\frac{1}{d+1}+\left(\sum_{h=0}^2 \stirling{d}{d-h}+\sum_{j\in [d-3]}\stirling{d}{j}\right)\frac{1}{2y}\right)\\
&=y^{d+1}\left(\sum_{j\in [d]} \stirling{d}{j}\right)\left(\frac{1}{d+1}+\frac{1}{2y}\right),
\end{align*}
and this shows the claim.
\end{proof}
By applying Facts \ref{fact1} and \ref{fact2} to \eqref{def_lambda_y} with $y\geq 2$ we obtain:
\begin{equation*}
\Lambda_{d}(y)=\frac{N(y)}{D(y)}\leq \frac{y^{d+1}\left(\sum_{j\in [d]} \stirling{d}{j}\right)\left(\frac{1}{d+1}+\frac{1}{2y}\right)}{y^{d+1}\left(\frac{1}{d+1}+\frac{1}{2y}\right)}=\sum_{j\in [d]} \stirling{d}{j}\leq \sum_{j\in [d]}\left(\frac{1}{j+1}+1\right)\stirling{d}{j}=\Lambda_d(1),
\end{equation*}
and this concludes the proof of the lemma. 
\end{proof}

\subsection{Missing proofs of Section \ref{sec6} (non-increasing latency functions)} \label{sec_A2}

\subsubsection{Proposition \ref{prop:sym:size1} does not extend}
The following example admits two pure Nash equilibria having distinct potentials, illustrating that Proposition \ref{prop:sym:size1}, which holds for non-decreasing latency functions, does not extend to non-increasing latency functions.   
\begin{example} \label{ex_non_incr}
Consider a symmetric and singleton congestion game 
with four players and two resources having the same latency function defined as follows. 
    $$ \latency(h) = \left\{
        \begin{array}{ll}
            2 & \text{, if } 1 \le  h < 4 \\
            1 & \text{, if } h = 4 \\
        \end{array} 
        \right. $$
A state  where all the players play the same resource is a Nash equilibrium with potential $7$. However, there exists another Nash equilibrium in which each resource is played by exactly two  players, 
and 
the potential is equal to $8$.
\end{example}

\subsubsection{Proof of Proposition \ref{prop1}}

\begin{proof} 
Fix a symmetric congestion game ${\cal G}$ with non-increasing latency functions.
Let us first observe that if two distinct strategies $a,b \in \S$ are actually played in a state $\ss$, then one of the following modifications of $\ss$ gives a new state $\ss'$ satisfying $\Phi_{\cal G}(\ss') \le \Phi_{\cal G}(\ss)$: either  all the players playing $a$  
change for $b$, or all the players playing $b$ 
change for $a$. (The other players stick to their strategy.)    

Let $i$ (resp., $j$) be a player such that $\strategy_i=a$ (resp., $\strategy_j=b$). If $\cost_i(\ss) \ge \cost_j(\ss)$, then all the players playing $a$ under $\ss$ can change their strategy for $b$, and their individual cost will not increase. Indeed, the latency functions being non-increasing, the new cost of the deviating players would be at most $\cost_j(\ss)$. If $\cost_i(\ss) < \cost_j(\ss)$, then all the players playing $b$ under $\ss$ can change their strategy for $a$, and their individual cost will decrease since it will be at most $\cost_i(\ss)$ (the fact that every latency function is non-increasing is used again). Since Rosenthal's function is an exact potential, we deduce that $\Phi_{\cal G}(\ss) \ge \Phi_{\cal G}(\ss')$, where $\ss'$ is the state obtained from $\ss$ by grouping the players of $a$ and $b$ either onto  $a$, or onto  $b$.      

We know from the above observation that there always exists a strategy profile $\ss^*$ that minimizes $\Phi_{\cal G}(\ss^*)$ in which all the players adopt the exact same strategy. From an algorithmic viewpoint, one can try every strategy $\strategy \in \S$, and retain the strategy profile $(\strategy,\ldots,\strategy)$ which minimizes Rosenthal's potential.  
\end{proof}

\subsubsection{Proof of Theorem \ref{thm2}}

\begin{proof}  Given a simple graph $G=(V,E)$, a vertex cover $C$ of $G$ is a subset of $V$ such that every edge $e \in E$ has at least one endpoint in $C$. The {\sc vertex cover} problem, which is known to be {NP}-hard \cite{RK72}, asks for a vertex cover of minimum cardinality.

Take a {\sc vertex cover} instance $G=(V,E)$ and create an instance ${\cal G}=\langle \N,\R,(\latency_r)_{r\in \R}\rangle$ of a singleton  non-increasing congestion game. Every edge $e \in E$ is associated with a player $p_e \in \N$. Every vertex $v \in V$ is associated with a resource $r_v \in \R$. For every edge $e=(v,w)$, the strategy space of player $p_e$ is defined as $\{r_v,r_w\}$. The latency function $\latency_r$ of every resource $r$ is such that $\latency_r(1)=1$, and $\latency_r(x)=0$ for all $x\ge 2$.        

We claim that $G$ admits a vertex cover $C$ of cardinality at most $k$ if, and only if, ${\cal G}$ has a state $\ss$ of potential at most $k$.

\noindent $(\Rightarrow)$ For each edge $e=(v,w)$, build a state $\ss$ so that player $p_e$ plays $v$ if $v \in C$, otherwise $p_e$ plays $w$.\footnote{If both $v$ and $w$ are in $C$, then choose one arbitrarily.} Since $C$ is a vertex cover, $\{v,w\} \cap C \neq \emptyset$. For each vertex that is played by at least one player in $\ss$, its contribution to $\Phi_{\cal G}(\ss)$ is 1, whereas unused vertices contribute 0 to $\Phi_{\cal G}(\ss)$. Therefore, $\Phi_{\cal G}(\ss) \le |C| \le k$.        

\noindent $(\Leftarrow)$ For any given state $\ss$ of ${\cal G}$, every player $p_e$ selects an adjacent vertex which covers edge $e$. Thus, $C_{\ss}:=\{v\in V : v=\strategy_{p_e} \mbox{ for some } e \in E\}$ describes a vertex cover of $G$. 
By the definitions of $\latency_r$ and $\Phi_{\cal G}$, $\Phi_{\cal G}(\ss)$ is exactly the number of vertices selected by at least one player. Therefore, $|C_\ss| = \Phi_{\cal G}(\ss) \le k$. 
\end{proof}

\subsubsection{Proof of Theorem \ref{thmHn}}

\begin{proof} 
The algorithm and its analysis follow the line of the greedy strategy used for approximating  the \textsc{set cover} problem (cf. 
\cite[Chapter 2]{VV01}). 

Consider a congestion game ${\cal G}$ where each resource $r$ is seen as a set that may cover player $i$ as long as $r \in \S_i$. At the beginning, the set of available 
players, denoted by $\hat \N$, is equal to $\N$. In addition,  the set of available resources, denoted by $\hat \R$, is equal to $\R$. At every step the {\em cost effectiveness} $\alpha_r$ of every resource $r \in \hat \R$ is computed. Let $\nu_r:=|\{i \in \hat \N : r \in \S_i \}|$ and define $\alpha_r$ as $\left( \sum_{t\in [\nu_r]} \latency_r(t) \right) / \nu_r$.  
Find the resource $r' \in \hat \R$ of minimum cost effectiveness, assign the players of the set $\{i \in \hat \N : r' \in \S_i \}$ to $r'$, and let $\alpha_{r'}$ be the {\em price} of these players. Both $\hat \R$ and $\hat \N$ are updated: $\hat \R \gets \hat \R \setminus \{r'\}$, and $\hat \N \gets \hat \N \setminus \{i \in \hat \N : r' \in \S_i \}$. The procedure is repeated  until every player is assigned to some resource. The resulting strategy profile is denoted by $\ss'$.

Note that since the latency functions are non-increasing, the cost effectiveness of a resource $r$ would not be smaller if we decided to assign it a proper subset of $\{i \in \hat \N : r \in \S_i \}$ instead of the entire set $\{i \in \hat \N : r \in \S_i \}$. For this reason,  every resource is assigned a group of players at most once during the execution of the algorithm.  

Let us analyze the worst-case value of the ratio $\Phi_{\cal G}(\ss')/\Phi_{\cal G}(\ss^*)$ where $\ss^*$ is a strategy profile of minimum potential.   
Rename the players in such a way that if player $i$ is assigned to a resource before player $i'$ during the execution of the algorithm, then $i<i'$ (break ties arbitrarily for the players assigned during the same step). Each assigned resource $r$ contributes to the potential of $\ss'$ by some amount $\sum_{\in [\nu_r]} \latency_{r}(t)$, and the sum of the prices of the players assigned to $r$ is precisely equal to $\sum_{\in [\nu_r]} \latency_{r}(t)$. Therefore, the sum of the prices over $\N$ is equal to $\Phi_{\cal G}(\ss')$.

Take any step of the execution of the algorithm where we assume that players $1$ to $i-1$ were previously assigned, and $\hat \N=\{i,i+1, \ldots, n\}$. The most cost effective resource is chosen to ``cover''  $i$ (and possibly other players of $\hat \N$). For every $r \in \hat \R$, 
let $\nu^*_r$ be the number of players of $\hat \N$ using $r$ under the optimal strategy profile $\ss^*$. By construction we have $\Psi:=\sum_{r\in \hat \R} \sum_{t\in [\nu^*_r]} \latency_r(t) \le \Phi_{\cal G}(\ss^*)$ (the inequality is due to the players in $\N \setminus \hat \N$ ). The minimum cost effectiveness of a resource in $\hat \R$ must be below $\Psi/|\hat \N|$, which is upper  bounded by $\Phi_{\cal G}(\ss^*)/|\hat \N|$. This means that the algorithm selects a resource with cost effectiveness at most $\Phi_{\cal G}(\ss^*)/|\hat \N|$, implying that the price of player $i$ is at most $\Phi_{\cal G}(\ss^*)/|\hat \N|=\Phi_{\cal G}(\ss^*)/(n-i+1)$.  
It follows that $\Phi_{\cal G}(\ss')=\sum_{i\in [n]} price(i)\le \sum_{i\in [n]} \Phi_{\cal G}(\ss^*)/(n-i+1)=\Phi_{\cal G}(\ss^*) H_{n}$.

One can see that the analysis is tight.  Suppose $\R$ consists of $n+1$ resources $r_0,r_1,\ldots,r_{n}$, and $\S_i=\{r_0,r_i\}$ for every player $i \in \N$. For $i\in [n]$, we have $\latency_{r_i}(1)=1/i$, and $\latency_{r_i}(t)=0$ for all $t>1$. Concerning $r_0$, we have $\latency_{r_0}(1)=1+\epsilon$ for some small but positive $\epsilon$, and $\latency_{r_0}(t)=0$ for all $t>1$. 
During each step $k \ge 1$ of the algorithm, the most cost effective resource is $r_{n+1-k}$ whose cost effectiveness is $1/(n+1-k)$; the cost effectiveness of $r_0$ is $(1+\epsilon)/(n+1-k)$. Thus, player $n+1-k$ is assigned to $r_{n+1-k}$ at each step $k \ge 1$. The potential of $\ss'$ is $\sum_{k\in [n]} 1/(n-k+1)=\sum_{k\in [n]} 1/k=H_{n}$. The state $\ss^*$ in which all the players are assigned to $r_0$ has potential $\sum_{k\in [n]} \latency_{r_0}(k)=\latency_{r_0}(1)=1+\epsilon$. The ratio $\Phi_{\cal G}(\ss')/\Phi_{\cal G}(\ss^*)=H_{n}/(1+\epsilon)$ tends to $H_{n}$ as $\epsilon$ goes to 0. 

It remains to prove that $H_{n}$ is the best 
approximation ratio unless P $=$ NP. An instance of \textsc{set cover} is given by a ground set $\Omega$ and a collection $\mathcal{C}$ of subsets of $\Omega$ such that $\bigcup_{C \in \mathcal{C}} C=\Omega$. Each set $C$ of $\mathcal{C}$ has a weight $w_C \in \RP$. The problem is to find $\mathcal{C}' \subseteq \mathcal{C}$ such that $\bigcup_{C \in \mathcal{C}'} C=\Omega$ and the total weight $\sum_{C \in \mathcal{C}'} w_C$ is minimum. Any $\rho$-approximation algorithm  
for minimizing Rosenthal's potential in singleton non-increasing congestion games gives a $\rho$-approximation algorithm for \textsc{set cover}. Take an instance of \textsc{set cover} and create an instance ${\cal G}$ of the game where each $C \in \mathcal{C}$ is associated with a resource $r_C$ whose latency function $\latency_{C}$ verifies $\latency_{C}(1)=w_C$, and $\latency_{C}(h)=0$ for all $h>1$. Each $i \in \Omega$ is associated with a player $i$ whose strategy space $\S_i$ is defined as $\{\{r_C\} : i \in C\}$. Then, every state $\ss$ of the game with potential $\Phi_{\cal G}(\ss)$ corresponds to a set cover $\mathcal{C}'$ of total weight $\Phi_{\cal G}(\ss)$, where $C \in \mathcal{C}' \Leftrightarrow (\exists i \in \N$ such that $\strategy_i=r_C)$. Provided that approximating \textsc{set cover} to within factor $(1-\varepsilon ) \ln (|\Omega|)$ is NP-hard for every $\varepsilon > 0$ \cite{DS14}, the ratio $H_{n}$ is best possible.
\end{proof}

\end{document}